\documentclass[journal,twoside,web]{ieeecolor}
\usepackage{lcsys}
\usepackage{cite}
\usepackage{amsmath,amssymb,amsfonts}
\usepackage{tikz}  
\usepackage{algorithmic}
\usepackage{graphicx}
\usepackage{textcomp}
\usepackage{graphicx}
\usepackage{caption}
\usepackage{makecell}
\usepackage{mwe}
\usepackage{hyperref}
\usepackage{comment}
\usepackage{fancyhdr}
\usepackage{xspace}
\usepackage{soul}
\usepackage{mathtools}
\usepackage{makecell}
\usepackage{cite}
\usepackage{nicefrac}
\usepackage{float}
\usepackage{wrapfig}
\usepackage{color}
\newtheorem{asm}{Assumption}
\newtheorem{prp}{Proposition}
\newtheorem{thm}{Theorem}
\newtheorem{crl}{Corollary}
\newtheorem{rem}{Remark}
\newcommand{\R}{\ensuremath{\mathbb{R}}}

\newcommand{\diag}{\ensuremath{\text{diag}}}

\makeatletter
\newcommand{\subalign}[1]{%
	\vcenter{%
		\Let@ \restore@math@cr \default@tag
		\baselineskip\fontdimen10 \scriptfont\tw@
		\advance\baselineskip\fontdimen12 \scriptfont\tw@
		\lineskip\thr@@\fontdimen8 \scriptfont\thr@@
		\lineskiplimit\lineskip
		\ialign{\hfil$\m@th\scriptstyle##$&$\m@th\scriptstyle{}##$\crcr
			#1\crcr
		}%
	}
}

\title{\LARGE \bf
Generalized Super-Twisting Observer for a class of interconnected nonlinear systems with uncertainties}

\author{Rania Tafat$^1$, Jaime A. Moreno$^2$ and Stefan Streif$^{1}$ % <-this % stops a space
\thanks{This work has been accepted in the IEEE Control Systems Letters,
10.1109/LCSYS.2025.3575432; ©2025 IEEE. Personal use of this material
is permitted. Permission from IEEE must be obtained for all other uses, in
any current or future media, including reprinting/republishing this material
for advertising or promotional purposes, creating new collective works,
for resale or redistribution to servers or lists, or reuse of any copyrighted
component of this work in other works.
}
\thanks{*This work was partially supported by UNAM-PAPIIT IN106323.}% <-this % stops a space
\thanks{*This research has been performed as part of the project ReSIDA-H2. This
project was funded by the European Social Fund Plus (ESF Plus) and the
Free State of Saxony.}
\thanks{$^{1}$ Technische Universit\"at Chemnitz, 09126 Chemnitz, Germany, Automatic Control and System Dynamics Lab; E-mail:
        {\tt\small \{rania.tafat, stefan.streif\}@etit.tu-chemnitz.de}}%
\thanks{$^{2}$ Eléctrica y Computación
Instituto de Ingeniería-UNAM
Universidad Nacional Autónoma de México; E-mail:
        {\tt\small JMorenoP@iingen.unam.mx}}
}

\begin{document}

\maketitle
\thispagestyle{empty}

%%%%%%%%%%%%%%%%%%%%%%%%%%%%%%%%%%%%%%%%%%%%%%%%%%%%%%%%%%%%%%%%%%%%%%%%%%%%%%%%
\begin{abstract}
The Generalized Super-Twisting Observer (GSTO) is extended for a strongly observable class of nonlinearly interconnected systems with bounded uncertainties/perturbations.
A nonsmooth strong Lyapunov function is used to prove the finite-time convergence of the proposed observer to the true system's trajectories, in the presence of the uncertainties. 
A case study on the interaction between two food production systems is presented, comparing the proposed observer with the High Gain observer. The results emphasize the critical role of the GSTO’s discontinuous term in achieving exact estimation.
\end{abstract}


%%%%%%%%%%%%%%%%%%%%%%%%%%%%%%%%%%%%%%%%%%%%%%%%%%%%%%%%%%%%%%%%%%%%%%
\section{Introduction} 
\IEEEPARstart{I}{nterconnected} nonlinear systems with uncertainties arise in various fields, including aerospace engineering, power systems, transportation networks and controlled agriculture environments.
The design of observers for such systems is, therefore, a topic of significant relevance and has been extensively studied \cite{wu2004decentralized,koo2015decentralized}.
A common approach is adressing the control and observer design simultaneously.
The control can steer the system toward a desired reference, where a Luenberger observer suffices for each subsystem, resulting in a decentralized observer \cite{dhbaibi2009h}. 
Alternatively, the controller  faciliates the convergence of a centralized observer by compensating the unknown uncertainties, as seen in the High Gain (HG) decentralized observer \cite{wu2004decentralized}, the adaptive neural network observer \cite{liu2024observer}, and the Luenberger observer aided by the fuzzy decentralized control \cite{li2019observer}.
Another approach focuses on the observer design of interconnected systems with uncertainties alone, using appropriate tools that guarantee convergence. 
For instance, in the presence of reconstructible unknown inputs and adequate assumptions, a distributed Luenberger based observer suffices \cite{chakrabarty2016distributed}. 
In a similar sense, if knowledge of structural characteristics of the uncertainties is available, the observer can compensate for them \cite{yan2003robust}.
Fuzzy observer techniques are also widely used for interconnected systems with uncertainties, but asymptotic convergence is guaranteed only in the absence of the latter \cite{koo2015decentralized}.
In fact, estimation errors provided by these methods are, at best, bounded or input-to-state (dynamically) stable with respect to the unknown uncertainties \cite{praly2001observers}.
This is a common drawback of continuous observers, which are often unable to compensate for system uncertainties and unknown inputs when their structural knowledge is missing \cite{moreno2013discontinuous}.

The classical sliding mode observer (SMO) was introduced to address this challenge, enabling exact state estimation despite bounded unknown uncertainties, which is possible thanks to the observer's discontinuous term \cite{utkin2013sliding}. 
Consequently,
 the SMO was also adopted for the robust observer design of interconnected systems with uncertainties, see e.g. \cite{farahani2022sliding} and references within. 
This conventional approach has however inherent limitations when it comes to interconnected systems and convergence guarantees are challenging to obtain.
Fortunately, several other discontinuous techniques have emerged, as the second order sliding mode observer (SOSMO) \cite{
levant1998robust}, the Super Twisting Algorithm (STA) \cite{davila2010variable} and the General Super Twisting Observer (GSTO) which, as its name suggests, generalizes STA for two dimensional systems exhibiting richer dynamical characteristics \cite{salgado2011generalized, moreno2013discontinuous}.
As for higher dimensional systems, the Multivariable Generalized Super-Twisting Algorithm (MGSTA) shows significant exploitation potential \cite{nagesh2014multivariable}.
Indeed, when the MGSTA consists of $k$ decoupled subsystems running in parallel, the origin of this algorithm remains finite-time stable despite unknown bounded uncertainties. 
A natural extension of this approach involves coupling these algorithms by treating their interconnections as perturbations, as explored in \cite{lopez2019generalised}. 
%This algorithm is widely exploited for the conceptualizing robust finite-time stabilizing control strategies for interconnected nonlinear systems \cite{koch2022conditioned, fang2015multivariable,dong2017adaptive} but still needs theoretical improvment for exploitation in the observer design community.
For observer design of linearly interconnected multi-agent systems without communication, one may refer to \cite{fomichev2024cascade}.
In addition, \cite{lopez2023quasicontinuous} also leverages MGSTA properties to obtain exact state estimation for a class of continuous, globally Lipschitz systems with an observable linear part in the presence of unknown inputs.
This letter extends that work by aiming to generalize the results  to a broader class of systems, by relaxing the continuous and global Lipschitz conditions and the necessity of an observable linear part. 
Instead, we focus on a class of nonlinearly interconnected nonlinear subsystems, which may exhibit discontinuities or multivariate interactions as discussed in Section \ref{section:Problem}. 
We begin by establishing that this class of systems is strongly observable in the presence of unknown inputs \cite{moreno2014dynamical} and introduce an observer plant for it using the GSTO. Section \ref{section:convergence} proves that the proposed plant drives the estimation error toward nonlinearly interconnected GSTAs whose origin is finite time stable, using a nonsmooth quadratic Lyapunov function \cite{moreno2008lyapunov}.
A motivating example is given in Section \ref{section:numerical_results} along with simulations of the proposed method  compared with a HG observer. Finally, Section \ref{section:conclusion} concludes this letter.
\section{Problem statement and main results}
\label{section:Problem} Consider a system  that can be transformed into the following interconnection of $N \geq 2$ subsystems
\begin{equation}
\begin{split}
\Sigma_1 &:
\begin{cases}
\dot{x}_{11} &= f_{11} (y,u) + g_1(y,u,t)x_{12} ,\\
\dot{x}_{12} &= f_{12}(x,u) + \delta_1(x,u,w,t), 
\end{cases} \\
\Sigma_2 &: \begin{cases}
    \dot{x}_{21} &= f_{21} (y,u,x_{12})+ g_2(y,u,t)x_{22}, \\
    \dot{x}_{22} &= f_{22}(x,u) + \delta_2(x,u,w,t),
\end{cases}\\
&\vdots\\
\Sigma_N &: \begin{cases}
    \dot{x}_{N1} &= f_{N1}\left(y,u,x_{12},x_{22}, \cdots, x_{(N-1)2}\right)  \\ & + g_N(y,u,t)x_{N2},\\
    \dot{x}_{N2} &= f_{N2}(x,u) + \delta_N(x,u,w,t),
\end{cases}\\
y &= \begin{bmatrix}
    x_{11} & x_{21} & \cdots & x_{N1}
\end{bmatrix}^\top,
\label{eq:system}
\end{split}
\end{equation}
with $i=1,\cdots,N$. Here $x \in \R^{2N}$ are the states, $u \in \R^m$ is a known input, $w \in \R^{N}$ are unknown inputs and $y \in \R^N$ is the measured output. 
$f_{i1}$ are known continuous functions and $f_{i2}$ are known possibly discontinous or multivalued functions. 
$\delta_{i}$ represent the unknown uncertain terms.
$g_i$ are known continuous bounded and strictly positive functions, i.e. without loosing generality, there exist positive constants 
$
0 < g_{i,m} \leq g_i( y, u,t) \leq g_{i,M}.
$
It is assumed that the system \eqref{eq:system} has a solution in the sense of \textit{Filippov} (see 1.2.1 in  \cite{bacciotti2005liapunov}).
\begin{rem}
  System \eqref{eq:system} represents a broad class of interconnected systems characterized by two types of nonlinear interconnections. The first state of each subsystem allows a cascaded nonlinear interconnection, while the second state accommodates an arbitrary nonlinear interconnection, which maybe unknown.
  Notably, the proposed observer in this work achieves finite-time convergence even when $\delta_{i}$ is non-vanishing, provided that Assumption \ref{asm:GSTO_interconnection} (see \eqref{eq:PertBound} below) holds. %Note also that the second state can be driven by an unknown but bounded perturbation term.
   \end{rem}
Recall that, according to \cite{moreno2014dynamical, hautus1983strong }, system \eqref{eq:system} is \emph{strongly observable} (roughly) if the state $x$ can be obtained from the time derivatives of the known input $u$ and of the measured output $y$, despite of the unknown input $\delta$.  %Before introducing the proposed interconnected observer, the next Proposition shows that system  in the sense of , i.e.  
%Before introducing the proposed interconnected observer, the main result of this letter, in the next Proposition we state an observability property of the class of systems \eqref{eq:system}.
\begin{prp}
    The uncertain interconnected system \eqref{eq:system} is \emph{globally strongly observable}.% in the sense of \cite{moreno2014dynamical, hautus1983strong }.
 \label{prp:observability}
\end{prp}
\begin{proof}
   Extending the proof for the case $N = 1$ in \cite{moreno2013discontinuous} consider the observability map of the system \eqref{eq:system}
   \begin{equation*}
    \begin{split}
    \begin{bmatrix}
            y \\ \dot{y}
    \end{bmatrix} &= \mathcal{O}\left(x,u\right) =  \\&\begin{bmatrix}
            x_{11} \\ 
            \vdots \\ x_{N1} \\ f_{11} (y,u) + g_1(y,u,t)x_{12} \\  \vdots \\
            f_{N1}\left(y,u,x_{12},x_{22}, \cdots, x_{(N-1)2}\right) + g_N(y,u,t)x_{N2} 
        \end{bmatrix}.
        \end{split}
    \end{equation*}
    Since $\mathcal{O}\left(x,u\right)$ does not depend on $\delta$ and is globally invertible in $x$ for every $u \in \R^m$, the conclusion follows.
\end{proof}
Although strong observability does not assure the existence of a (continuous) unknown input observer \cite{moreno2014dynamical,hautus1983strong}, we propose the following discontinuous one, based on the GSTA \cite{salgado2011generalized, moreno2013discontinuous}, which can converge for bounded $\delta_i$ %unknown inputs 

\begin{equation}
\begin{split}
\hat{\Sigma}_1 &: \begin{cases}
    \dot{\hat{x}}_{11} &= -\gamma l_{11}g_1(y,u,t)\phi_{11}(e_{11}) + f_{11}(y,u) \\&+ g_1(y,u,t)\hat{x}_{12}, \\
    \dot{\hat{x}}_{12} &= -\gamma^2 l_{12} g_1(y,u,t) \phi_{12}(e_{11}) + f_{12}(\hat{x},u),
\end{cases}\\
%\hat{\Sigma}_2 &: \begin{cases}
 %   \dot{\hat{x}}_{21} &= - \gamma l_{21} g_2(y,u,t) \phi_{21}(e_{21}) + f_{21}(y,u,\hat{x}_{12}) \\ & + g_2(y,u,t)\hat{x}_{22}\\
    %\dot{\hat{x}}_{22} &= - \gamma ^2 l_{22} g_2(y,u,t) \phi_{22}(e) + f_{22}(\hat{x},u)
%\end{cases} \\
&\vdots\\
\hat{\Sigma}_N &: \begin{cases}\dot{\hat{x}}_{N1} &= - \gamma l_{N1} g_N(y,u,t) \phi_{N1}(e_{N1})\\ &+ f_{N1}\left(y,u,\hat{x}_2, \cdots, \hat{x}_{(N-1)2} \right), \\
\dot{\hat{x}}_{N2} &= - \gamma^2 l_{N2}g_N (y,u,t) \phi_{N2}(e_{N1}) + f_{N2}(\hat{x},u),
\end{cases}
\end{split}
\label{eq:observer}
\end{equation}
where $l_{ij}>0$ are positive gains to be designed and
\begin{equation}
\begin{aligned}
e_{ij} &= \hat{x}_{ij} - x_{ij}, \\
\phi_{i1}\left(z \right) &= \mu_{i1}\left|z\right|^{\frac{1}{2}} \operatorname{sign}\left(z\right) +\mu_{i2} z, \, \mu_{i1}, \mu_{i2} > 0,\\
\phi_{i2}\left(z\right) &= \frac{\mu_{i1}^2}{2} \operatorname{sign}\left(z\right)+\mu_{i1} \frac{3}{2}\mu_{i2}\left|z\right|^{\frac{1}{2}} \operatorname{sign}\left(z\right)+\mu_{i2}^2 z.
\end{aligned}
\end{equation}  
Solutions of \eqref{eq:observer} are understood in the sense of \textit{Filippov} \cite{bacciotti2005liapunov}.
To ensure the convergence of \eqref{eq:observer} to the true trajectories of \eqref{eq:system}, an assumption on the interconnections between the subsystems is necessary and is introduced next.
\begin{asm}
There exist nonnegative real constants $\alpha_{ij}$, $\tilde{\alpha}_{ij}$, $\beta_{ij}$ and $\tilde{\beta}_{ij}$ for which 
    \begin{equation}
    \begin{split}
    |\rho_{i1} (x,e,y,u)| & \leq \sum\limits_{j=1}^{N} \tilde{\alpha}_{ij} |e_{j1}| + \tilde{\beta}_{ij} |e_{j2}|,\\
     |\rho_{i2} (x,e,y,u)| &\leq \alpha_{i0} +\sum\limits_{j=1}^{N} \alpha_{ij} |e_{j1}| + \beta_{ij} |e_{j2}|,
          \end{split}
         \label{eq:PertBound}
    \end{equation}
    where \begin{equation}
        \begin{split}
        \rho_{11} (x,e,y,u)&= 0,\\
       \rho_{i1}(x,e,y,u) &= f_{i1}
        \left(x_{(i-1)2}+e_{(i-1)2}, \cdots, x_{12} + e_{12} ,y,u \right) \\ &- f_{i1}\left(x_{(i-1)2}, \cdots,x_{12},y,u \right), \, \forall i\geq 2, \\
          \rho_{i2}(x,e,y,u)  &= f_{i2}\left( x + e, u \right) - f_{i2}(x,u) - \delta_{i}, \, \forall i \geq 1.
        \end{split}
    \end{equation}
    \label{asm:GSTO_interconnection}
  Due to the structural assumption in \eqref{eq:system}—namely, that the first channels are interconnected in cascade—we have
\begin{equation}
\tilde{\alpha}_{ij} = \tilde{\beta}_{ij}=0, \forall j \geq i,  \forall i=1,\cdots,N.
\label{eq:CascadeInt}
\end{equation}
\end{asm}
Building on this and Proposition \ref{prp:observability}, the GSTO result for the class of systems \eqref{eq:system} can be formulated, leading to our main result, which is presented in the next theorem.
\begin{thm}
  Given Assumption \ref{asm:GSTO_interconnection}, and if 
  $l_{ij}>0$ then there exists a value $\gamma_0>0$ such that for every $\gamma \geq \gamma_0$,
  the observer \eqref{eq:observer} converges to the true plant \eqref{eq:system} in finite-time.  
\label{thm:GSTO_interconnected}
\end{thm}
   \begin{rem}
   Note that selecting $\mu_{i1} = 0$ observer \eqref{eq:observer} is continuous and becomes a High-Gain Observer (HGO). However, the HGO cannot achieve exact state estimation in the presence of non vanishing uncertainties/perturbations $\delta_i \neq 0$, so that Theorem \ref{thm:GSTO_interconnected} requires gains $\mu_{ij}$ to be strictly positive. For such systems a continuous observer can, at best, ensure input-to-state stability with respect to the unknown input. %
     This underscores the necessity of discontinuity in the observer to achieve the desired finite-time convergence.
   \end{rem}
   It can also be shown that the estimation error remains bounded for bounded additive measurement noise. 
   The next section discusses and proves these claims. 
\section{Convergence proof using a quadratic Lyapunov function}
\label{section:convergence}
In this section, the proof of Theorem \eqref{thm:GSTO_interconnected} is established. In what follows we omit the dependencies of $g_i$, $\rho_{i}$ and $\delta_i$ for better readability.
The estimation error dynamics is given by
    \begin{equation}
    \begin{split}
        \dot{e}_{i1} &=  - \gamma l_{i1}g_i\phi_{i1}(e_{i1}) + g_1 e_{i2} + \rho_{i1},\\
        \dot{e}_{i2} &= - \gamma^2 l_{i2} g_i \phi_{i2} (e_{i1}) +  \rho_{i2},\, \, \forall i \in \{ 1, \cdots, N\}. 
    \end{split}
    \label{eq:ErrorEstimacion}
    \end{equation}
%\cite{moreno2008lyapunov,moreno2012strict,davila2009optimal,moreno2011lyapunov} 
In \cite{moreno2008lyapunov,moreno2011lyapunov} a quadratic
Lyapunov function (LF), that is continuous but not Lipschitz continuous,
has been introduced for the analysis of the convergence and robustness
properties of Super-Twisting-like algorithms. This LF
is quadratic not in the state vector, but in a vector 
\begin{equation}
\epsilon_{i}^{\top}=\varphi_{i}^{\top}\left(e_{i}\right)=
\begin{bmatrix}
\phi_{i1}\left(e_{i1}\right)\  & e_{i2}
\end{bmatrix}\ ,\label{Vect_chi}
\end{equation}
where $\varphi_{i}$ is a global homeomorphism \cite{moreno2011lyapunov}. 
To take the time derivative
of the LF it is necessary to calculate the time derivative of $\epsilon_i$,
that is given by (where it exists)
\begin{align*}
\dot{\epsilon}_i & =\phi_{i1}^{\prime}\left(e_{i1}\right)
\begin{bmatrix}
-l_{i1}\gamma g_i \phi_{i1}\left(e_{i1}\right)+g_i e_{i2}+\rho_{i1} \\
-l_{i2}\gamma ^{2} g_i \phi_{i1}\left(e_{1}\right)
\end{bmatrix} +
\begin{bmatrix}
0 \\
\rho_{i2}
\end{bmatrix} \\
& =\phi_{i1}^{\prime}\left(e_{i1}\right)\left\{ g_i \left(A_{0}-\Gamma L_{i}C_{0}\right)\epsilon_i + b_1 \rho_{i1}\right\}  + b_0 \rho_{i2} \ ,
\end{align*}
with
\[
A_{0}= \begin{bmatrix}
0 & 1 \\
0 & 0
\end{bmatrix},\, 
b_{1}= \begin{bmatrix}
1 \\
0
\end{bmatrix},\, 
b_{0}= \begin{bmatrix}
0 \\
1
\end{bmatrix},\, 
L_{i}= \begin{bmatrix}
l_{i1} \\
l_{i2}
\end{bmatrix}, \]
\[
C_{0} = \begin{bmatrix}
1 & 0
\end{bmatrix}, \,
\Gamma = \diag\{\gamma, \gamma^2\},
%\begin{bmatrix}
%\gamma & 0 \\
%0 & \gamma^2
%\end{bmatrix},
\]
where we have used the error equation \eqref{eq:ErrorEstimacion} and the relation
\[
\phi_{i2}\left(e_{i1}\right) = \phi_{i1}^{\prime}\left(e_{i1}\right) \phi_{i1}\left(e_{i1}\right)
\]
that can be easily established. Note that the characteristic polynomial of the matrix $\left(A_{0}-\Gamma L_{i}C_{0}\right)$
is
\begin{align*}
p\left(s\right) & %=\det\left(s\mathbb{I}-\left(A_{0}-\Gamma L_{i}C_{0}\right)\right)
=s^{2}+\gamma l_{i1}s+\gamma^{2}l_{i2} %\\&
=\left(s-\gamma \lambda_{i1}\right)\left(s-\gamma \lambda_{i2}\right),
\end{align*}
where $\lambda_{i1}$, $\lambda_{i2}$ are the eigenvalues of the (Hurwitz)
matrix $A_{il}=\left(A_{0}-L_{i}C_{0}\right)$, i.e matrix $\left(A_{0}-\Gamma L_{i}C_{0}\right)$
with $\gamma =1$. This shows that the eigenvalues of $\left(A_{0}-\Gamma L_{i}C_{0}\right)$
are $\gamma \lambda_{i1}$, $\gamma \lambda_{i2}$, multiples of the eigenvalues
of $\left(A_{0}-L_{i}C_{0}\right)$. 

Similarly to the proof method for High Gain Observer
\cite{khalil2002nonlinear}, we introduce here a further change of
variables
\[
\xi_{i} = \gamma \Gamma^{-1} \epsilon_{i}
= \begin{bmatrix}
\epsilon_{i1}, &
\frac{1}{\gamma}\epsilon_{i2}    
\end{bmatrix}^{\top},
%=\left[\begin{array}{c}
%\frac{\theta}{\gamma}\epsilon_{i1}\\
%\frac{\theta}{\gamma^{2}}\epsilon_{i2}
%\end{array}\right],
\]
we obtain
(since $\Gamma^{-1}A_{0}\Gamma=\gamma A_{0}$ and $C_{0}\Gamma=\gamma C_{0}$)
\begin{align*}
\dot{\xi}_i 
& = \phi_{i1}^{\prime}\left(e_{i1}\right)
\left\{ \gamma g_i \left(A_{0}- L_{i}C_{0}\right) \xi_i + b_1 \rho_{i1} \right\}  
+ b_0  \frac{1}{\gamma}  \rho_{i2}.
\end{align*}
Using for each subsystem a quadratic LF \cite{moreno2011lyapunov}
\begin{equation}
V_i\left(\xi_i\right)=\xi_i^{\top}P_i\xi_i,
\end{equation}
where $P_i$ is the unique, symmetric and positive definite ($P_i=P_i^{\top}>0$)
solution of the Algebraic Lyapunov Equation 
\[
\left(A_{0}-L_{i}C_{0}\right)^{\top}P_i+P_i\left(A_{0}-L_{i}C_{0}\right)=-Q_i,
\]
%for $Q_i=Q_i^{\top}>0$, an arbitrary positive definite and symmetric matrix \cite{khalil2002nonlinear}.
for an arbitrary $Q_i=Q_i^{\top}>0$ \cite{khalil2002nonlinear}.
The derivative of $V_i$ along the solutions of the corresponding error is %equation (almost everywhere) is given by
\begin{equation}
\begin{split}
\dot{V}_i
&= \phi_{i1}^{\prime}\left(e_{i1}\right)
\left\{ - \gamma g_i \xi_i^{\top} Q_i \xi_i  + 2 \rho_{i1} b_1^{\top} P_i\xi_i  
\right\} +  
\frac{2}{\gamma} \rho_{i2} b_0^{\top}P_i\xi_i.
\end{split}
\end{equation}
We consider first the term in the curly brackets. It satisfies, for any $0 < \eta <1$,
\begin{equation}
\begin{split}
 &- \gamma  g_i \xi_i^{\top} Q_i \xi_i  + 2 \rho_{i1}  b_1^{\top} P_i\xi_i 
  \leq - \eta \gamma  g_{i,m} \lambda_{\min}\left\{ Q_i\right\} 
\left\Vert \xi_i\right\Vert^{2} \\ &- \left(1-\eta\right) \gamma  g_{i,m} \lambda_{\min}\left\{ Q_i\right\} 
\left\Vert \xi_i\right\Vert^{2} + 2 \left\Vert P_i \right\Vert \left\Vert \xi_i \right\Vert  |\rho_{i1}|, \\
& \leq - \eta \gamma  g_{i,m} \lambda_{\min}\left\{ Q_i\right\} 
\left\Vert \xi_i\right\Vert^{2}, \, 
\\ & \forall \left\Vert \xi_i\right\Vert  \geq \frac{2 \lambda_{\max}\left\{ P_i\right\} }{\left(1-\eta\right) \gamma }g_{i,m} \lambda_{\min}\left\{ Q_i\right\} |\rho_{i1}| ,
\end{split}
\end{equation}
where $\lambda_{\min}\left\{ Q_i \right\} $ is the minimal eigenvalue
of $Q_i$, $\left\Vert \xi_i\right\Vert $ is the Euclidean norm of $\xi_i$
and $\left\Vert P_i\right\Vert =\lambda_{\max}\left\{ P_i \right\} $ is
the induced (Euclidean) norm of matrix $P_i$ and $g_{i,m}$ is the minimal value of $g_i$. %
Noting that $\phi_{i1}^{\prime}\left(e_{i1}\right)\geq0$,
since $\phi_{i1}\left(e_{i1}\right)$ is monotone increasing, we 
further obtain
\begin{equation}
\begin{split}
\label{eq:DerV_1}
\dot{V}_i &\leq - \eta \gamma  g_{i,m} \lambda_{\min}\left\{ Q_i\right\} \phi_{i1}^{\prime}\left(e_{i1}\right)
\left\Vert \xi_i\right\Vert^{2}  + \frac{2 \left\Vert P_i \right\Vert}{ \gamma}  \left\Vert \xi_i \right\Vert |\rho_{i2}|, \\
&\forall \left\Vert \xi_i\right\Vert  \geq \frac{2 \lambda_{\max}\left\{ P_i\right\}|\rho_{i1}|}{\left(1-\eta\right) \gamma g_{i,m} \lambda_{\min}\left\{ Q_i\right\}}.
\end{split}
\end{equation}
Recall the standard inequality for quadratic forms
\[
\lambda_{\min}\left\{ P_i \right\} \left\Vert \xi_i \right\Vert _{2}^{2}\leq V_i\left(\xi_i\right) = \xi_i^{\top}P_i\xi_i\leq\lambda_{\max}\left\{ P_i\right\} \left\Vert \xi_i\right\Vert _{2}^{2}\ ,
\]
where
\begin{equation*}
\begin{split}
\left\Vert \xi_i\right\Vert _{2}^{2} & =\xi_{i1}^{2}+\xi_{i2}^{2}=\phi_{i1}^{2}\left(e_{i1}\right)+\frac{1}{\gamma^{2}}e_{i2}^{2},\\
 & =\left(\mu_{i1}^{2}\left\vert e_{i1}\right\vert +2\mu_{i1}\mu_{i2}\left\vert e_{i1}\right\vert ^{\frac{3}{2}}+\mu_{i2}^{2}\left\vert e_{i1}\right\vert ^{2}\right)+\frac{1}{\gamma^{2}}e_{i2}^{2},
 \end{split}
\end{equation*}
is the Euclidean norm of $\xi_i$. Besides, note that the inequality 
\begin{equation}
\begin{split}
\left\vert e_{i1}\right\vert ^{\frac{1}{2}}&\leq\frac{1}{\mu_{i1}}\left\vert \phi_{i1}\left(e_{i1}\right)\right\vert \leq \frac{1}{\mu_{i1}}\left\Vert \xi_i\right\Vert %\\ &
\leq \frac{1}{\mu_{i1}\lambda_{\min}^{\frac{1}{2}}\{P_i\}}  V^{\frac{1}{2}}\left(\xi_i\right)
\end{split}\label{xpineq}
\end{equation}
is satisfied for any $\mu_{i1}>0$, and therefore
\[
-\frac{1}{\left\vert e_{i1}\right\vert ^{\frac{1}{2}}}\leq- \frac{\mu_{i1}}{\left\Vert \xi_i\right\Vert } \leq -\mu_{i1}\lambda_{\min}^{\frac{1}{2}}\{P_i\}V_i^{-\frac{1}{2}} \left(\xi_i\right).
\]
Since 
\begin{equation}
\phi_{i1}^{\prime}\left(e_{i1}\right)=\frac{1}{2}\mu_{i1}\frac{1}{\left\vert e_{i1}\right\vert ^{\frac{1}{2}}}+ \mu_{i2} 
\end{equation}
it follows from \eqref{eq:DerV_1} that, $\forall \left\Vert \xi_i\right\Vert  \geq \frac{2 \lambda_{\max}\left\{ P_i\right\} }{\left(1-\eta\right) \gamma g_{i,m} \lambda_{\min}\left\{ Q_i\right\}} |\rho_{i1}|$ 
\begin{equation}
  \begin{split}
\dot{V}_i & \leq  
- \frac{1}{2}\eta g_{i,m} \lambda_{\min}\left\{ Q_i\right\}
\mu_{i1}^2\gamma\left\Vert \xi_i\right\Vert 
\\ &- \eta \gamma  g_{i,m} \lambda_{\min}\left\{ Q_i\right\} \mu_{i2} \left\Vert \xi_i\right\Vert^{2}
  + \frac{2}{\gamma} \left\Vert P_i \right\Vert \left\Vert \xi_i \right\Vert  |\rho_{i2}|, \label{eq:DerV_2} \\
\dot{V}_i & \leq  
- \frac{\eta \mu_{i1}^2\gamma g_{i,m} \lambda_{\min}\left\{ Q_i\right\}}{2\lambda_{\max}^{\frac{1}{2}}\left\{ P_i\right\}} V_i^{\frac{1}{2}}\left(\xi_i\right) 
\\ &- \frac{\eta \gamma  g_{i,m} \lambda_{\min}\left\{ Q_i\right\} \mu_{i2}}{\lambda_{\max}\left\{ P_i\right\}}  V_i\left(\xi_i\right)
  +  \frac{2}{\gamma} \left\Vert P_i \right\Vert \left\Vert \xi_i \right\Vert |\rho_{i2}|.    
  \end{split}
\end{equation}
In the absence of interconnection terms and unknown inputs, i.e. for $\rho_{i1}=\rho_{i2}=0$, $V_i\left(\xi_i\left(t\right)\right) < - c_i V_i^{\frac{1}{2}}(\xi_i(t))$ with $c_i> 0$, and the origin $\xi_i=0$ is finite-time stable according to Theorem 4.2 in \cite{bhat2000finite}. This can be interpreted as an extension of the MGSTA convergence result \cite{lopez2015qualitative, lopez2019generalised} when $g_i \neq 1$.  

To account for the interconnections, consider for the estimation error of the full system the following LF candidate
\begin{equation}
    V\left(\xi\right) = \sum_{i=1}^{N}  V_i\left(\xi_i\right).
    \label{eq:LFC}
\end{equation}
Using \eqref{eq:DerV_2}, the time derivative of the candidate LF satisfies $\forall \left\Vert \xi_i\right\Vert  \geq \frac{2 \lambda_{\max}\left\{ P_i\right\}}{\left(1-\eta\right) \gamma g_{i,m} \lambda_{\min}\left\{ Q_i\right\}} |\rho_{i1}|$
\begin{align*}
    &\dot{V}= \sum_{i=1}^{N}  \dot{V}_i 
    \leq - \sum_{i=1}^{N} 
\frac{1}{2} \eta g_{i,m} \lambda_{\min}\left\{ Q_i\right\}\mu_{i1}^2\gamma \left\Vert \xi_i\right\Vert  
\\ &- \sum_{i=1}^{N}  \eta \gamma  g_{i,m} \lambda_{\min}\left\{ Q_i\right\} \mu_{i2} \left\Vert \xi_i\right\Vert^{2}
  %\\&
  + 2 \sum_{i=1}^{N}  \frac{\left\Vert P_i \right\Vert}{\gamma} \left\Vert \xi_i \right\Vert |\rho_{i2}|\, .
\end{align*}
Note that, since
\begin{align*}
\left\vert e_{i1} \right\vert & \leq \frac{1}{\mu_{i2}} \left\vert \phi_{i1}\left(e_{i1}\right) \right\vert = \frac{1}{\mu_{i2}} \left\vert \epsilon_{i1} \right\vert = \frac{1}{\mu_{i2}} \left\vert \xi_{i1} \right\vert, \\
\left\vert e_{i2} \right\vert & = \left\vert \epsilon_{i2} \right\vert = \gamma\left\vert \xi_{i2} \right\vert,
\end{align*}
and using the bounds on the terms $\rho_{i1},\rho_{i2}$ in \eqref{eq:PertBound}, %can be given as
%        \begin{equation}
%        \begin{split}
%         |\rho_{i1}| & \leq \sum_{j=1}^{N} \frac{\gamma}{\theta} \left( \frac{\tilde{\alpha}_{ij}}{ \mu_{j2}} \left\vert \xi_{j1} \right\vert + \tilde{\beta}_{ij} \gamma\left\vert \xi_{j2} \right\vert \right), \\
%         |\rho_{i2}| &\leq \alpha_0 + \sum_{j=1}^{N} \frac{\gamma}{\theta} \left( \frac{\alpha_{ij}}{\mu_{j2}} \left\vert \xi_{j1} \right\vert + \beta_{ij} \gamma\left\vert \xi_{j2} \right\vert \right),
%         \end{split}
%        \end{equation}
we obtain %and thus
    \begin{equation}
    \begin{split}
        \frac{1}{\gamma}|\rho_{i1}|
        & \leq \frac{1}{\gamma}\sum_{j=1}^{N} \left( \frac{\tilde{\alpha}_{ij}}{ \mu_{j2}} + \tilde{\beta}_{ij} \gamma \right) \left\Vert \xi_j \right\Vert, \\
         \frac{1}{\gamma}|\rho_{i2}|
         & \leq \frac{1}{\gamma}\alpha_0 + 
         \frac{1}{\gamma }
         \sum_{j=1}^{N}  \left( \frac{\alpha_{ij}}{\mu_{j2}} + \beta_{ij} \gamma \right)
         \left\Vert \xi_j \right\Vert.
         \end{split}
    \end{equation}
Since $\tilde{\alpha}_{ii}=\tilde{\beta}_{ii}=0$, we arrive at
\begin{equation}
\begin{split}
    \dot{V}
         &\leq - \sum_{i=1}^{N} 
\left( \frac{\gamma}{2} \eta g_{i,m} \lambda_{\min}\left\{ Q_i\right\}\mu_{i1}^2 - 2 \lambda_{\max}\left\{ P_i\right\} \frac{\alpha_0}{\gamma } \right)
\left\Vert \xi_i\right\Vert  
\\ &- \sum_{i=1}^{N}  \eta \gamma  g_{i,m} \lambda_{\min}\left\{ Q_i\right\} \mu_{i2} \left\Vert \xi_i\right\Vert^{2}  \\
& + \sum_{i=1}^{N} \sum_{j=1}^{N} \frac{2  \lambda_{\max}\left\{ P_i\right\} }{\gamma }
          \left( \frac{\alpha_{ij}}{\mu_{j2}} + \beta_{ij} \gamma \right)
         \left\Vert \xi_i \right\Vert
         \left\Vert \xi_j \right\Vert,
\end{split}
\label{eq:der_V4}
\end{equation}
which is satisfied for the subset $\Omega$ of the state space, described by the following set of inequalities: 
$\forall  i=1,\cdots,N$
\begin{equation}
\begin{split}
    \left\Vert \xi_i\right\Vert  &\geq \frac{2 \lambda_{\max}\left\{ P_i\right\} }{\left(1-\eta\right) g_{i,m} \lambda_{\min}\left\{ Q_i\right\}} \frac{1}{\gamma } \sum_{j=1,j\neq i}^{N} \left( \frac{\tilde{\alpha}_{ij}}{ \mu_{j2}} + \tilde{\beta}_{ij} \gamma \right) \left\Vert \xi_j \right\Vert
    \\ &\geq \frac{2 \lambda_{\max}\left\{ P_i\right\} }{\left(1-\eta\right) \gamma  g_{i,m} \lambda_{\min}\left\{ Q_i\right\}} |\rho_{i1}|. 
    \end{split}
    \label{eq:NegSet}
\end{equation}
Using the classical inequality, for all $x,y\geq 0$
\begin{equation}
    xy \leq \frac{1}{2}x^2 + \frac{1}{2}y^2,
\end{equation}
on the last term of \eqref{eq:der_V4}, we obtain
\begin{equation}
\begin{split}
\dot{V}   & \leq - \sum_{i=1}^{N} 
\left( \frac{\gamma }{2} \eta g_{i,m} \lambda_{\min}\left\{ Q_i\right\}\mu_{i1}^2- 2 \lambda_{\max}\left\{ P_i\right\} \frac{\alpha_0}{\gamma } \right)
\left\Vert \xi_i\right\Vert \\
& - \sum_{i=1}^{N} 
    \gamma  \left\{ 
     \eta g_{i,m} \lambda_{\min}\left\{ Q_i\right\} \mu_{i2} \right. \\ & 
     \left. - \sum_{j=1}^{N} 
\left[ \frac{ \lambda_{\max}\left\{ P_i\right\} }{\gamma ^2}
    \left( \frac{\alpha_{ij}}{\mu_{j2}} + \beta_{ij} \gamma \right) \right. \right. \\ & \left. \left. +
    \frac{ \lambda_{\max}\left\{ P_j\right\} }{\gamma ^2}
          \left( \frac{\alpha_{ii}}{\mu_{i2}} + \beta_{ii} \gamma \right)
    \right] 
    \right\}
         \left\Vert \xi_i \right\Vert^2.
         \end{split}
\end{equation}

It is clear that there is a value $\gamma_{\min}$ such that
\begin{align}
    \dot{V} & \leq - \sum_{i=1}^{N} c_i\left(\gamma\right) V_i^{\frac{1}{2}}\left(\xi_i\right) 
- \sum_{i=1}^{N} \tilde{c}_i\left(\gamma\right)  V_i\left(\xi_i\right), \forall \xi \in \Omega
\label{eq:LyapIneq}
\end{align}
with $c_i\left(\gamma\right) > 0$ and $\tilde{c}_i\left(\gamma\right)>0$ for every $\gamma \geq \gamma_{\min}$.

This implies that the trajectories of the observer are ultimately and uniformly bounded, i.e. they converge to a neighborhood of the origin $e=0$ and remain there for all future times \cite{lopez2020finite}. In fact, the smallest level set of the LF that is contained in $\Omega$ is a positively invariant set.

The previous result is true for arbitrary interconnections of the subobservers without self interconnection $\tilde{\alpha}_{ii}=\tilde{\beta}_{ii}=0$. 
For the cascade interconnection in the first channel in Assumption \ref{asm:GSTO_interconnection}, a stronger conclusion is possible. 
In this case $\rho_{11}=0$ and then \eqref{eq:DerV_2} for $i=1$ becomes 
\begin{equation}
\begin{split}
&\dot{V}_1  \leq  
- \left\{ \frac{\gamma}{2}\eta g_{1,m} \lambda_{\min}\left\{ Q_1\right\}
\mu_{11}^2 
- \frac{2 \left\Vert P_1 \right\Vert \alpha_0}{\gamma}  \right.
  \\ & \left. - \frac{2 \left\Vert P_1 \right\Vert }{\gamma}
    \sum_{j=2}^{N} \left( \frac{\alpha_{ij}}{\mu_{j2}} + \beta_{ij} \gamma \right)
         \left\Vert \xi_j \right\Vert \right\} \left\Vert  \xi_1 \right\Vert \\
         & - \gamma \left[ \eta g_{1,m} \lambda_{\min}\left\{ Q_1\right\} \mu_{12} -
         \frac{2 \left\Vert P_1 \right\Vert }{\gamma^2} \left( \frac{\alpha_1}{\mu_{12}} + \beta_1 \gamma \right) \right]
         \left\Vert \xi_1 \right\Vert^2,
         \label{eq:derV1}
  \end{split}
\end{equation}
which is satisfied for all $\left\Vert \xi_1\right\Vert \geq 0$. 
Since the terms $\left\Vert \xi_j \right\Vert$ remain bounded for every $\gamma \geq \gamma_{\min}$, then there exists $\bar{\mu}_{1j}> 0$ such that $\forall \mu_{1j} \geq \bar{\mu}_{1j}$, \eqref{eq:derV1} becomes $\dot{V}_1 < - c_1 V^{1/2}, \, \, \forall \| \xi_{1} \| > 0$ and for $c_1 > 0$. 
This means that there exists $t_1 \geq t_0$ such that $\| e_{1j}(t_1) \| = 0$ and because the interconnection $\rho_{i1}$ are cascaded, $\rho_{21}$ vanishes for all $t \geq t_1$.
Applying an identical reasoning for $i=2$, and since $\rho_{21}(t) = 0 , \forall t \geq t_1$, we can conclude from \eqref{eq:DerV_2} that $\xi_{2}$ converges to the origin in finite-time which ensures that $\rho_{31}$ vanishes and succesively show that $\xi_3, \cdots, \xi_{N}$ converge to the origin in finite-time. Thus, we conclude that the observer converges to the true state trajectories within a finite-time, thereby completing our proof.
\section{Numerical results}
\begin{figure}[b!]
 \centering
 \includegraphics[width=1\linewidth]{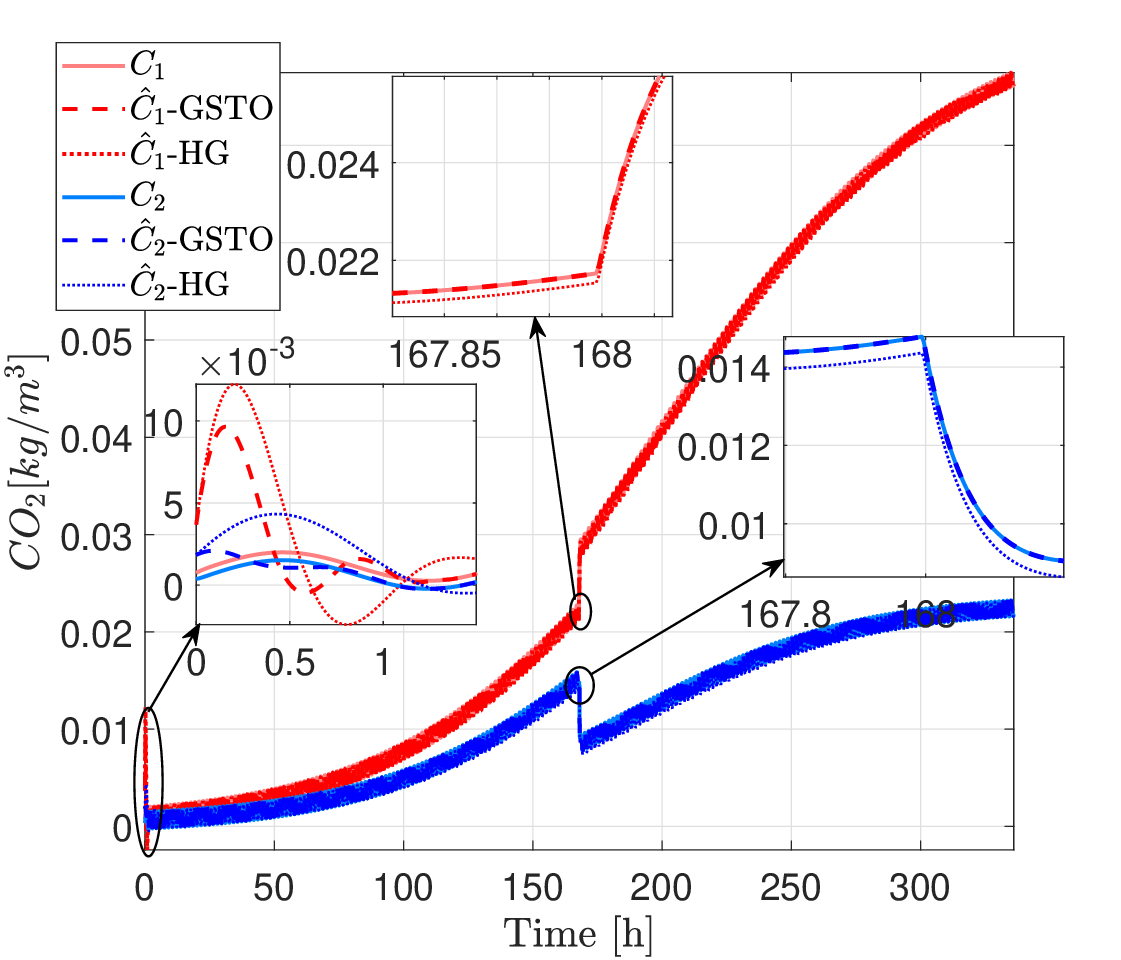}
    \caption{$CO_2$ measurements tracking of the interconnected larvae PUs showing exact convergence of proposed GSTO.}
    \label{fig:CO2}
\end{figure}
\begin{figure}[h!]
    \centering
   \includegraphics[width=1\linewidth]{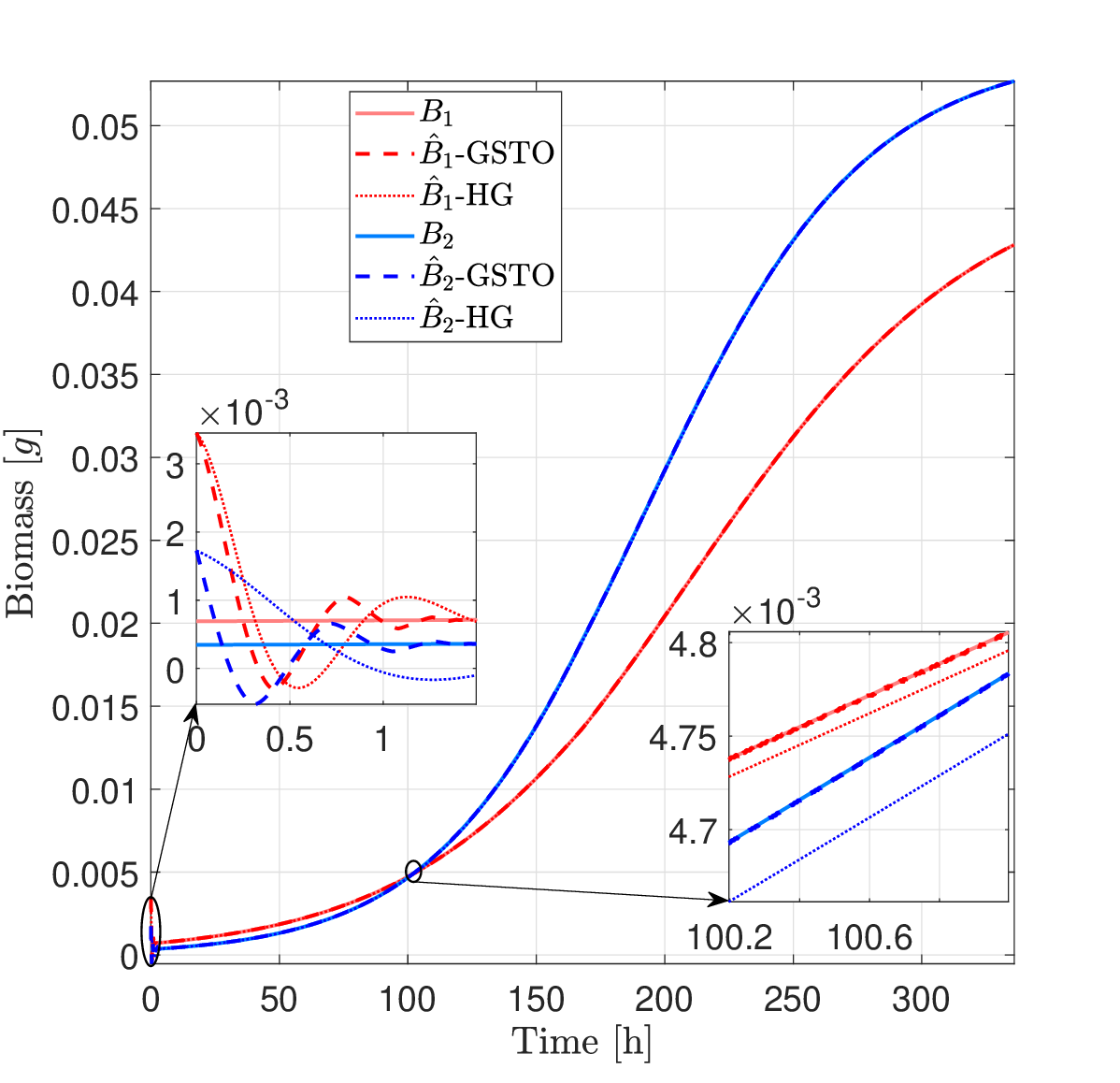}
    \caption{Dry biomass per larval estimation.}
    \label{fig:Biomass}
\end{figure}
 \label{section:numerical_results}
In order to illustrate our results, consider a common problem in food production and insects farming: to estimate dry biomass per larval $B_i$ from $CO_2$ and $O_2$ measurements, $C_i$ and $O_i$, respectively. 
We focus on $N$ larvae production units (PU) interconnected via valves $u_{v,i}$, as it is modeled in \cite{al2022centralized}, to which we incorporate the air rate function from \cite{padmanabha2020comprehensive}.
Let us consider that the $(i+1)$-th PU growing conditions are better than the $i$-th PU, e.g. $T_{i+1} > T_{i}$, s. t. $T_i$ is the temperature of the growing medium $i$, which leads the insects to crawl from the $i$th to the $(i+1)$th PU.
The model becomes
\begin{equation}
 \begin{split}
     \dot{C}_1 &= \alpha_7u_{v1} \left( C_{2} - C_1 \right) + \alpha_9 u_{o1}\left( C_{o} - C_1 \right)  + L_1 \alpha_{15}r_{A1}B_1, \\
     \dot{B}_1 &= \alpha_{20}\alpha_1(1 - \alpha_{17})r_{A1}B_1^2 + \alpha_{2}\alpha_{21} r_{A1}r_{T1}B_1 - \kappa u_{v1}B_1,\\
     &\vdots\\
     \dot{C}_N &= \alpha_7u_{vN} \left( C_{N-1} - C_N \right) + \alpha_9 u_{oN}\left( C_{o} - C_N \right) \\& + L_N \alpha_{15}r_{AN}B_N + \kappa u_{vN} \alpha_{15} r_{AN}B_N,\\
     \dot{B}_N &= \alpha_{20}\alpha_1(1 - \alpha_{17})r_{AN}B_N^2 \\&+ \alpha_{2}\alpha_{21} r_{AN}r_{TN}B_N + \kappa u_{v(N-1)}B_{N-1},\\
 \end{split}  
 \label{eq:example}
\end{equation}
where $C_{o}$ is the outside $CO_2$ and is measured, $r_{Ai} = \frac{O_i}{O_i + C_i}$, $L_i$ the number of larvae of the $i$th PU, while $r_{Ti}$ is unknown, since it depends on the unknown medium temperature. 
Moreover, we assume all parameters known except for $\alpha_1$ and $\kappa$, as they both depend on unmeasurable growing conditions.
Consequently, $\dot{B}_i$ is completely unknown. 
\begin{rem}
  Interconnected food production systems frequently belong to the class of systems described by \eqref{eq:system} due to their rate functions, which are nonlinear expressions dependent on states, inputs, and unknown signals. Moreover, these environmental rate functions are inherently bounded by nature, making Assumption \ref{asm:GSTO_interconnection} naturally satisfied.
\end{rem}
\begin{figure}
    \centering
\includegraphics[width=1\linewidth]{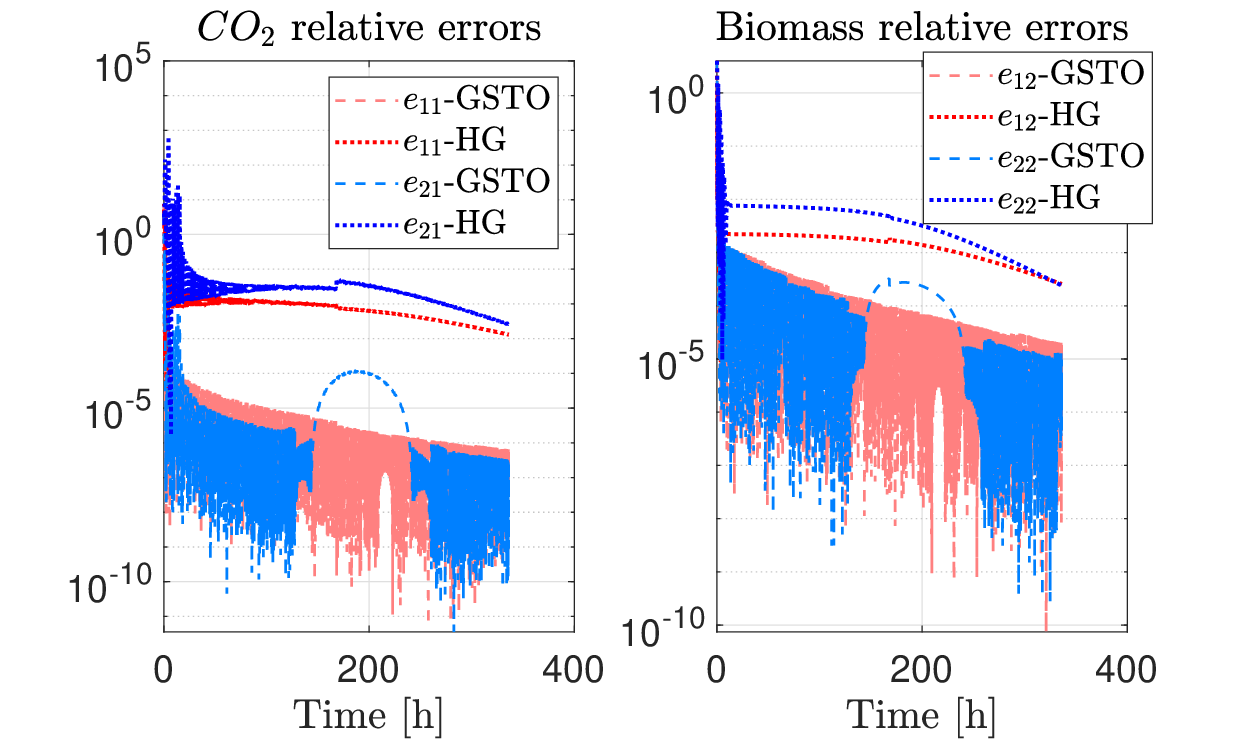}
    \caption{Relative estimation errors with perfect measurements.}
    \label{fig:relative error}
\end{figure}
\begin{figure}
    \centering   
    \vspace{5pt}
    \includegraphics[width=1\linewidth]{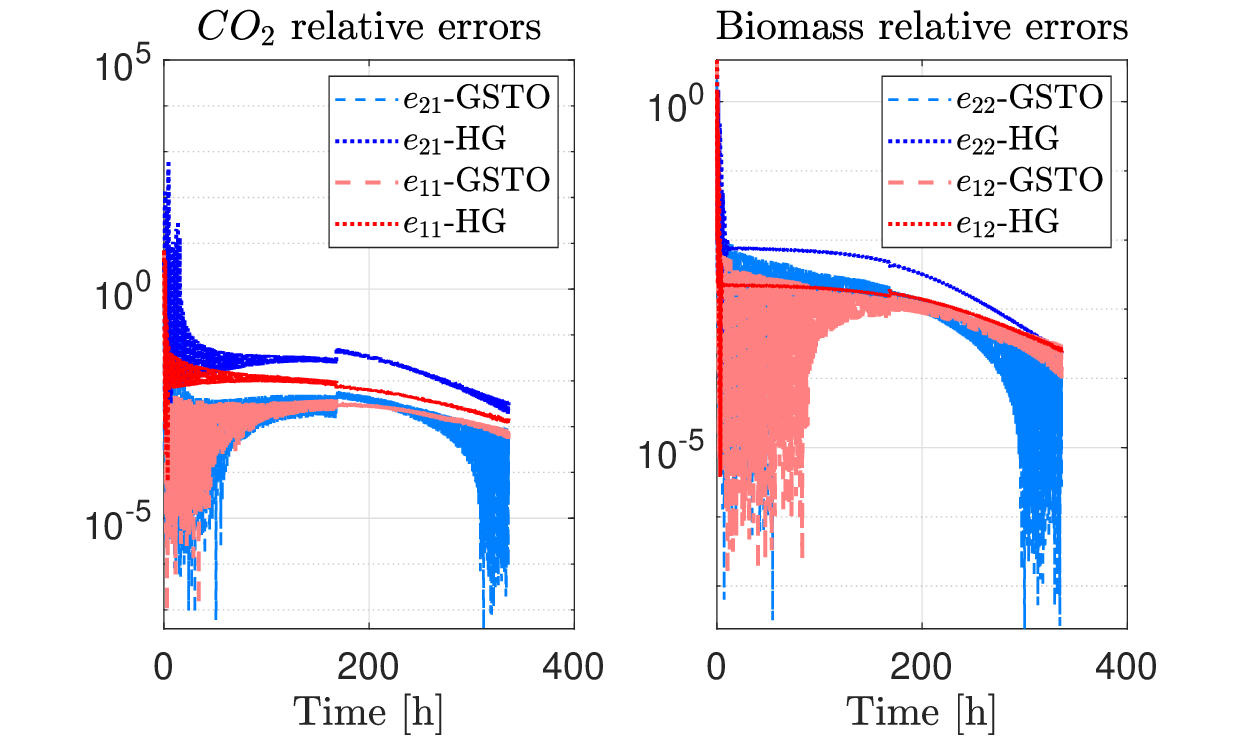}
    \caption{Relative estimation errors in the noisy case.}
    \label{fig:errors_noisy}
\end{figure}
Figures \ref{fig:CO2}, \ref{fig:Biomass} and \ref{fig:relative error} show-cast a comparison between the proposed GSTO \eqref{eq:observer} and a HGO on system \eqref{eq:example} for $N=2$ and parameters taken from \cite{al2022centralized}, simulated for two weeks with a time step of $10[s]$. We set $T_1 = 35^\circ C,  \, T_2 = 40^\circ C, \, \kappa = 10^{-6}$, $r_T$ is the modified log10 function \cite{padmanabha2020comprehensive}, $u_v = 0.4$ when $t\leq 7$[day] and $u_v = 0$ elsewhere. While $u_{o} = 0.4 \forall t$ and $C_{o} = 9.1167 \operatorname{sin}(\frac{t}{900})\times10^{-4}$. The true initial conditions are $\begin{bmatrix}
    7.37\times10^{-4} & 6.9\times10^{-4} & 3.69\times10^{-4} & 3.45\times10^{-4}
\end{bmatrix}^\top$ and the initial guess of each observer is $5$ times the truth. The GSTO's parameters are $l_{i1} =  1.1,\, l_{i1} = 3, \, \gamma_1 = 0.1, \, \gamma_2 = 0.5, \mu = \begin{bmatrix}
    0.03 & 1 \\ 0.01 & 1
\end{bmatrix}$. 
The HGO is obtained by setting $\mu_{i1} = 0$.
The jump in the $CO_2$ concentration in Figure \ref{fig:CO2} corresponds to the closing of the valve at day $7$, where the interactions stop.
Notice that this HG fails to accurately track the $CO_2$ measurements, and increasing its gain is not a feasible option, as it would cause an excessive peaking. This occurs due to the fast dynamics of the measurements and the HG observer's lack of information on biomass dynamics. This is precisely why the discontinuous term is essential in such case; it prevents the peaking phenomenon by enforcing a fast and precise estimation, as illustrated in the figures.
Additive noise was introduced as $ y_1 = C_1 + 0.03\operatorname{sin}(t),\, 
        y_2 = C_2 + 0.1\operatorname{sin}(t)$. Figure \ref{fig:errors_noisy} shows relative errors, where the interconnected GSTO continues to outperform the HGO despite bounded noise. Note that $\gamma$ must be made to balance convergence speed and noise sensitivity, as excessively high values may amplify noise.
\section{Conclusion}
We have introduced an observer that achieves exact state estimation for a strongly observable class of nonlinearly interconnected systems with uncertainties. 
To maintain the strong observability of the class, the interconnection on the measured states are required to be cascaded while the ones affecting the unmeasured states can be arbitrary, multivariate and unknown. 
The observer is designed so that the estimation error dynamics form a network of nonlinearly interconnected GSTAs.
To establish convergence, we extended existing GSTO results by employing a nonsmooth quadratic LF, allowing us to account for nonlinear and partially unknown interconnections between subsystems. Our analysis demonstrates that, in the presence of general nonlinear interconnections and bounded uncertainties, the observer's discontinuity is crucial for ensuring convergence—something unattainable with its continuous counterpart. 
This conclusion is further validated through numerical simulations of an interconnection of larvae production units.
Additionally, our analysis indicates that if the requirement on the cascaded interconnections for this class of systems is lifted, then the proposed GSTO is finite time input-to-state stable, ensuring that estimation errors remain uniformly bounded around the origin. A promising direction for future research is to establish conditions under which a strongly observable system can be transformed into the proposed class.
%Future research should investigate alternative interconnection structures, such as feedback, which may yield even stronger results.
\label{section:conclusion}

\addtolength{\textheight}{-1cm}   % This command serves to balance the column lengths

%%%%%%%%%%%%%%%%%%%%%%%%%%%%%%%%%%%%%%%%%%%%%%%%%%%%%%%%%%%%%%%%%%%%%%%%%%%%%%%%

%%%%%%%%%%%%%%%%%%%%%%%%%%%%%%%%%%%%%%%%%%%%%%%%%%%%%%%%%%%%%%%%%%%%%%%%%%%%%%%%

%%%%%%%%%%%%%%%%%%%%%%%%%%%%%%%%%%%%%%%%%%%%%%%%%%%%%%%%%%%%%%%%%%%%%%%%%%%%%%%%

%%%%%%%%%%%%%%%%%%%%%%%%%%%%%%%%%%%%%%%%%%%%%%%%%%%%%%%%%%%%%%%%%%%%%%%%%%%%%%%%

\phantomsection
\addcontentsline{toc}{chapter}{Bibliography} 
%\nocite{*} 
\bibliographystyle{acm}
\bibliography{bib}

\begin{thebibliography}{10}

\bibitem{al2022centralized}
{\sc Al~Khatib, M., Hempel, A.-J., Padmanabha, M., and Streif, S.}
\newblock Centralized optimization of resource routing in interconnected food
  production units with harvesting events.
\newblock {\em IFAC-PapersOnLine 55}, 7 (2022), 322--327.

\bibitem{bacciotti2005liapunov}
{\sc Bacciotti, A., and Rosier, L.}
\newblock {\em Liapunov functions and stability in control theory}.
\newblock Springer Science \& Business Media, 2005.

\bibitem{bhat2000finite}
{\sc Bhat, S.~P., and Bernstein, D.~S.}
\newblock Finite-time stability of continuous autonomous systems.
\newblock {\em SIAM Journal on Control and Optimization 38}, 3 (2000),
  751--766.

\bibitem{chakrabarty2016distributed}
{\sc Chakrabarty, A., Sundaram, S., Corless, M.~J., Buzzard, G.~T., {\.Z}ak,
  S.~H., and Rundell, A.~E.}
\newblock Distributed unknown input observers for interconnected nonlinear
  systems.
\newblock In {\em American Control Conference\/} (2016), pp.~101--106.

\bibitem{davila2010variable}
{\sc D{\'a}vila, A., Moreno, J.~A., and Fridman, L.}
\newblock Variable gains super-twisting algorithm: A {L}yapunov based design.
\newblock In {\em Proc. of the 2010 American Control Conference\/} (2010),
  pp.~968--973.

\bibitem{dhbaibi2009h}
{\sc Dhbaibi, S., Tlili, A.~S., Elloumi, S., and Braiek, N.~B.}
\newblock $h_{infty}$ decentralized observation and control of nonlinear
  interconnected systems.
\newblock {\em ISA Transactions 48}, 4 (2009), 458--467.

\bibitem{farahani2022sliding}
{\sc Farahani, A.~V., and Abolfathi, S.}
\newblock Sliding mode observer design for decentralized multi-phase flow
  estimation.
\newblock {\em Heliyon 8}, 2 (2022).

\bibitem{fomichev2024cascade}
{\sc Fomichev, V., and Samarin, A.}
\newblock Cascade super-twisting observer for linear multiagent systems without
  communication.
\newblock {\em Differential Equations 60}, 2 (2024), 247--258.

\bibitem{hautus1983strong}
{\sc Hautus, M.~L.}
\newblock Strong detectability and observers.
\newblock {\em Linear Algebra and its Applications 50\/} (1983), 353--368.

\bibitem{khalil2002nonlinear}
{\sc Khalil, H.~K.}
\newblock {\em Nonlinear systems}, vol.~3.
\newblock Prentice hall Upper Saddle River, NJ, 2002.

\bibitem{koo2015decentralized}
{\sc Koo, G.~B., Park, J.~B., and Joo, Y.~H.}
\newblock Decentralized sampled-data fuzzy observer design for nonlinear
  interconnected systems.
\newblock {\em IEEE Transactions on Fuzzy Systems 24}, 3 (2015), 661--674.

\bibitem{levant1998robust}
{\sc Levant, A.}
\newblock Robust exact differentiation via sliding mode technique.
\newblock {\em automatica 34}, 3 (1998), 379--384.

\bibitem{li2019observer}
{\sc Li, Y.-X., Tong, S., and Yang, G.-H.}
\newblock Observer-based adaptive fuzzy decentralized event-triggered control
  of interconnected nonlinear system.
\newblock {\em IEEE Transactions on Cybernetics 50}, 7 (2019), 3104--3112.

\bibitem{liu2024observer}
{\sc Liu, S., Xu, N., Zhao, N., and Zhang, L.}
\newblock Observer-based optimal fault-tolerant tracking control for
  input-constrained interconnected nonlinear systems with mismatched
  disturbances.
\newblock {\em Optimal Control Applications and Methods 45}, 6 (2024),
  2572--2595.

\bibitem{lopez2015qualitative}
{\sc L{\'o}pez-Caamal, F., and Moreno, J.~A.}
\newblock Qualitative differences of two classes of multivariable
  super-twisting algorithms.
\newblock In {\em 54th IEEE Conference on Decision and Control\/} (2015),
  pp.~5414--5419.

\bibitem{lopez2019generalised}
{\sc L{\'o}pez-Caamal, F., and Moreno, J.~A.}
\newblock Generalised multivariable supertwisting algorithm.
\newblock {\em International Journal of Robust and Nonlinear Control 29}, 3
  (2019), 634--660.

\bibitem{lopez2023quasicontinuous}
{\sc L{\'o}pez-Caamal, F., and Moreno, J.~A.}
\newblock A quasicontinuous multivariable super-twisting observer for 2 n
  states systems with lipschitz nonlinearities.
\newblock {\em International Journal of Robust and Nonlinear Control 33}, 15
  (2023), 8992--9017.

\bibitem{lopez2020finite}
{\sc Lopez-Ramirez, F., Efimov, D., Polyakov, A., and Perruquetti, W.}
\newblock Finite-time and fixed-time input-to-state stability: Explicit and
  implicit approaches.
\newblock {\em Systems \& Control Letters 144\/} (2020), 104775.

\bibitem{moreno2011lyapunov}
{\sc Moreno, J.~A.}
\newblock Lyapunov approach for analysis and design of second order sliding
  mode algorithms.
\newblock In {\em Sliding modes after the first decade of the 21st century:
  State of the art}. Springer, 2011, pp.~113--149.

\bibitem{moreno2013discontinuous}
{\sc Moreno, J.~A.}
\newblock On discontinuous observers for second order systems: properties,
  analysis and design.
\newblock {\em Advances in Sliding Mode Control: Concept, Theory and
  Implementation\/} (2013), 243--265.

\bibitem{moreno2008lyapunov}
{\sc Moreno, J.~A., and Osorio, M.}
\newblock A {L}yapunov approach to second-order sliding mode controllers and
  observers.
\newblock In {\em 47th IEEE Conference on Decision and Control\/} (2008),
  pp.~2856--2861.

\bibitem{moreno2014dynamical}
{\sc Moreno, J.~A., Rocha-C{\'o}zatl, E., and Wouwer, A.~V.}
\newblock A dynamical interpretation of strong observability and detectability
  concepts for nonlinear systems with unknown inputs: application to
  biochemical processes.
\newblock {\em Bioprocess and Biosystems Engineering 37\/} (2014), 37--49.

\bibitem{nagesh2014multivariable}
{\sc Nagesh, I., and Edwards, C.}
\newblock A multivariable super-twisting sliding mode approach.
\newblock {\em Automatica 50}, 3 (2014), 984--988.

\bibitem{padmanabha2020comprehensive}
{\sc Padmanabha, M., Kobelski, A., Hempel, A.-J., and Streif, S.}
\newblock A comprehensive dynamic growth and development model of hermetia
  illucens larvae.
\newblock {\em Plos One 15}, 9 (2020), e0239084.

\bibitem{praly2001observers}
{\sc Praly, L.}
\newblock On observers with state independent error {L}yapunov function.
\newblock {\em IFAC Proceedings Volumes 34}, 6 (2001), 1349--1354.

\bibitem{salgado2011generalized}
{\sc Salgado, I., Chairez, I., Moreno, J., Fridman, L., and Poznyak, A.}
\newblock Generalized super-twisting observer for nonlinear systems.
\newblock {\em IFAC Proceedings Volumes 44}, 1 (2011), 14353--14358.

\bibitem{utkin2013sliding}
{\sc Utkin, V.~I.}
\newblock {\em Sliding modes in control and optimization}.
\newblock Springer Science \& Business Media, 2013.

\bibitem{wu2004decentralized}
{\sc Wu, Q., Jiang, L., and Wen, J.}
\newblock Decentralized adaptive control of interconnected non-linear systems
  using high gain observer.
\newblock {\em International Journal of Control 77}, 8 (2004), 703--712.

\bibitem{yan2003robust}
{\sc Yan, X.-G., Lam, J., and Xie, L.}
\newblock Robust observer design for non-linear interconnected systems using
  structural characteristics.
\newblock {\em International Journal of Control 76}, 7 (2003), 741--746.

\end{thebibliography}

\end{document}